\documentclass[10pt]{article}

\usepackage[margin=1in]{geometry}
\usepackage{graphicx}
\usepackage{amsthm}
\usepackage{amsmath}
\usepackage{enumitem}
\usepackage{setspace}
\usepackage{times}
\usepackage{setspace}
\usepackage{sectsty}

\sectionfont{\fontsize{10pt}{11pt}\selectfont\MakeUppercase}

\setlist[itemize]{topsep=0pt,partopsep=0pt, itemsep=0pt, parsep=0pt}
\setlist[enumerate]{partopsep=0pt, itemsep=0pt, parsep=0pt}

\theoremstyle{definition}
\newtheorem{definition}{Definition}[section]

\theoremstyle{plain}
\newtheorem{lemma}{Lemma}[section]

\graphicspath{{./}}

\begin{document}

    \title{A Polynomial Time Algorithm for 3SAT} 
    \author{Robert Quigley}

    \begin{abstract}
        It is shown that any two clauses in an instance of 3SAT sharing the same terminal which is positive in one clause and negated in the other can imply a new clause composed of the remaining terms from both clauses. Clauses can also imply other clauses as long as all the terms in the implying clauses exist in the implied clause. It is shown an instance of 3SAT is unsatisfiable if and only if it can derive contradicting 1-terminal clauses in exponential time. It is further shown that these contradicting clauses can be implied with the aforementioned techniques without processing clauses of length 4 or greater, reducing the computation to polynomial time. Therefore there is a polynomial time algorithm that will produce contradicting 1-terminal clauses if and only if the instance of 3SAT is unsatisfiable. Since such an algorithm exists and 3SAT is NP-Complete, P = NP.
    \end{abstract}
    
    \maketitle 
    
    \section{Introduction}

        This section introduces the 3SAT problem, the implications of
        solving it in polynomial time, and the structure of the paper.

        As seen in \cite{10.1145/800157.805047},
        The boolean satisfiability problem is given as a set of terminals, $x_1, x_2, ..., x_n$,
        each of which can be assigned a value of True or False,
        combined by logical AND operators, logical OR operators, and negations.
        The problem is to determine whether or not there exists an assignment
        for each terminal that allows the instance to evaluate to True.
        As seen in \cite{Karp1972}, the satisfiability problem
        with exactly 3 literals per clause is NP-Complete. This is the same 
        as the boolean satisfiability problem, but it is presented such that
        a clause contains exactly three terminals combined with logical OR operators
        and possibly negations, and each clause is combined with logical AND operators.
        The idea of NP-completeness shows that if one NP-complete problem can be solved
        in polynomial time, then all problems in the class NP can be solved in polynomial time.
        In other words, P = NP. 

        The paper is structured as follows:
        \begin{enumerate}
            \item Introduction
            \item Standard Definitions
            \begin{itemize}
                \item Terms relating to the 3SAT problem
            \end{itemize}
            \item Algorithm-Specific Definitions
            \begin{itemize}
                \item Terms relating to specific aspects of 3SAT regularly referenced in this paper
            \end{itemize}
            \item Reformatting and Processing
            \begin{itemize}
                \item Defines how clauses and instances of the 3SAT problem will appear within the paper
            \end{itemize}
            \item Lemmas
            \begin{itemize}
                \item A list of lemmas and their proofs pertaining to the algorithm
            \end{itemize}
            \item Algorithm
            \begin{itemize}
                \item A step-by-step description of the algorithm to solve 3SAT in polynomial time
            \end{itemize}
            \item Time Complexity Analysis
            \begin{itemize}
                \item An analysis of the time complexity of the algorithm
            \end{itemize}
            \item Proof of Correctness
            \begin{itemize}
                \item A proof showing the algorithm works for every instance of 3SAT
            \end{itemize}
            \item Conclusion
            \item References
        \end{enumerate}
        
    \section{Standard Definitions}

    \begin{definition}
        Terminals: symbols used in the 3SAT problem that can be assigned a value
        of either 0 or 1, True or False, or any other binary assignment. They usually
        take the form $x_i$ where $i$ is a natural number.
    \end{definition}
    \begin{definition}
        Terms: terminals in either the positive or negated form that appear in a clause.
        If a term is positive, then the value of the terminal will be the 
        same as the value of the term. If a term is negated, then the value
        of the terminal will be the opposite of the value of the term.
    \end{definition}
    \begin{definition}
        Clauses: a set of terms combined by logical OR operators.
        Clauses usually take the form 
        $(x_i \lor \neg x_j \lor x_k)$ where $x_i \neq x_j \neq x_k$.
    \end{definition}
    \begin{definition}
        (3SAT) Instance: a set of any number of clauses combined by logical
        AND operators. Terminals may not repeat within a clause
        \footnote{See Lemma 5.3}
        , but they are free to repeat between clauses. 
        Instances usually take the form:

        $(x_i \lor \neg x_j \lor x_k) \land (x_l \lor x_m \lor x_n)$.
    \end{definition}
    \begin{definition}
        Assignment: A list of values in which each value represents either True
        or False such that each item in the list corresponds to a terminal and 
        all terminals are assigned a value.
    \end{definition}
    \begin{definition}
        Satisfying Assignment: An assignment, $A$, is said to satisfy the instance if applying $A$ will make the instance evaluate to True.
    \end{definition}
    \begin{definition}
        Satisfiable: an instance is satisfiable iff there exists a satisfying assignment.
    \end{definition}
    \begin{definition}
        Unsatisfiable: an instance is unsatisfiable iff there does not exist a satisfying assignment.
    \end{definition}

    \section{Algorithm-Specific Definitions}
    \begin{definition}
        Blocking an Assignment: An assignment, $A$, is said to be blocked by a clause, $C$ if, 
        given an instance containing $C$, there is no way that $A$ allows $C$ to evaluate to True, and
        thus there is no way $A$ allows the instance to evaluate to True.
    \end{definition}
    \begin{definition}
        Implication: A clause, $C$, is said to imply another clause, $D$, if all 
        assignments blocked by $D$ are also blocked by $C$.
    \end{definition}
    \begin{definition}
        Given Clauses: clauses which were given in the original instance.
    \end{definition}
    \begin{definition}
        Derived or Implied Clauses: clauses which are implied by the clauses in the original instance.
    \end{definition}
    \begin{definition}
        k-terminal (k-t) clause: a clause is described as a k-terminal or a k-t clause
        if there are $k$ terminals in the clause.
    \end{definition}
    \begin{definition}
        Reduction: A special type of implication in which the implied 
        clause is shorter than the implying clause(s).
    \end{definition}
    \begin{definition}
        Expansion: A special type of implication in which the 
        implied clause is longer than the implying clause(s) and all terms in the implying
        clause exist in the implied clause.
    \end{definition}
    \begin{definition}
        Contradicting 1-terminal Clauses: A set of two clauses are considered to be contradicting 1-terminal clauses if (1) they are both of length 1, (2) they contain the same terminal, and (3) the terminal is positive in one clause and negated in the other.
    \end{definition}

    \section{Reformatting and Processing}

    Since there are a lot of constant characteristics about an instance of 3SAT, 
    we can remove most of them to allow ourselves to focus only on what changes
    from instance to instance. A list of unchanging characteristics follows:
    \begin{itemize}
        \item the symbol $x$
        \item logical AND operators
        \item logical OR operators
    \end{itemize}
    
    The only difference between instances, therefore, is the subscript of the terminal.
    
    Additionally, the following items will be changed to improve compatibility
    with the Python programming language wherein an instance is expressed as
    a list of lists and each inner list represents a clause:
    \begin{itemize}
        \item parentheses will become square brackets
        \item negation symbols will become minus signs
        \item an instance may be surrounded with square brackets to show it is a list of lists
    \end{itemize}
   
    For example, the instance:
    
    $(\neg x_a \lor x_b \lor x_c) \land (x_a \lor x_d \lor x_e)$

    will be written as:

    $[[-a, b, c], [a, d, e]]$

    Instances of 3SAT will further be processed by removing any clauses that do not block
    any assignments. Since this algorithm relies on implications of new clauses, if we ever
    come across a clause that blocks no assignments, then by definition it will not be able
    to imply any additional clauses that block at least one assignment.

    As such, we will ignore any clauses given or derived that are described in Lemma 5.3

    \section{Lemmas}

    \begin{lemma}
        A clause can block an assignment.
    \end{lemma}
    \begin{proof}
        Recall an assignment is blocked if it does not allow the clause to evaluate to True.

        Since terms in a clause are combined by logical OR operators, a clause cannot evaluate to True if all terms in the clause evaluate to False.

        A term evaluates to False if it's either (1) negated and the terminal's value is True or (2) positive and the terminal's value is False.

        Given a clause, we know there are some number of unique terminals.
        
        Want to find an assignment where all the terms are assigned a value of False.

        Since an assignment exists for all possible values for each terminal, then there exists an assignment such that all the terms in the clause evaluate to False.

        Since all the terms evaluate to False and are combined by logical OR operators, the clause will evaluate to False.

        Since the instance is composed of clauses combined by logical AND operators and one clause evaluates to False, then the entire instance evaluates to False.

        Since the instance evaluates to False, the assignment cannot satisfy the instance.
    \end{proof}

    \begin{lemma}        
        For a given instance with $n$ terminals, there are $2^n$ possible assignments.
    \end{lemma}
    \begin{proof}
        An assignment for this instance consists of $n$ values, each with two possible values, True or False.

        Therefore, there are $2^n$ possible assignments.
    \end{proof}

    \begin{lemma}
        If a clause contains the same terminal in its negated and positive form, it will not block any assignments.
    \end{lemma}
    \begin{proof}

        Consider a clause containing the same terminal in both the positive and negative form.

        We know that terminal must either be True or False. Consider both cases:

        That terminal's value is True: The positive form of the terminal will be True and the clause will evaluate to True.

        That terminal's value is False: The negated form of the terminal will be True and the clause will evaluate to True.

        Since the clause evaluates to True in every case, there will never be an assignment with which it is impossible to make this clause evaluate to False.

        In other words, this clause blocks no assignments.
    \end{proof}

    \begin{lemma}
        Each clause of length $k$ blocks $2^{n-k}$ assignments.
    \end{lemma}
    \begin{proof}
        Consider the generic $k$-terminal clause, $C$, in an instance with $n$ terminals.

        As seen in Lemma 5.1, this blocks all assignments where all of the terms evaluate to False.

        Since the values for $k$ terminals are set, there are $n-k$ terminals left whose values could be True or False.

        Since an assignment exists for every possible way to assign values to each terminal, we know an assignment exists for every possible way to assign a value for these $n-k$ terminals.

        There are two possible ways to assign values to each of these $n-k$ terminals so there are $2^{n-k}$ unique assignments blocked by $C$.
    \end{proof}

    \begin{lemma}
        For any clause, $C$, if we select a terminal, $t$, that's not in $C$ then half of the assignments blocked by $C$ will assign True to $t$ and the other half will assign 
        False to $t$.
    \end{lemma}
    \begin{proof}
        We have a clause, $C$, of fixed yet arbitrary length, $k$:

        $[a, b, c, ...]$

        Now we select a terminal, $t$, that's not in $C$.

        Want to show half of the assignments blocked by $C$ assign True to $t$ and the other half assign False to $t$.

        We know there will be no overlap between these assignments because a single assignment cannot assign both the values True and False to the same terminal.

        Now we just have to show that exactly half of the assignments are blocked by assigning either True or False to $t$.

        By Lemma 5.4, $C$ blocks $2^{n-k}$ assignments. 

        If we fix the value of $t$, then there are only $n-k-1$ terminals whose values could be 0 or 1. 
        
        Since there are two choices per terminal and there are $n-k-1$ terminals, then there are $2^{n-k-1}$ assignments blocked by $C$ where the value of $t$ is fixed.

        Divide to get the ratio of the number of assignments blocked by adding $t$ to the number of assignments blocked by $C$ without $t$:
        
        $2^{n-k-1}/2^{n-k}$

        = $2^{n - k - 1 - (n - k)}$

        = $2^{-1}$

        = $1/2$

        This shows that half of the assignments blocked by $C$ assign a fixed value to $t$.

        Since there are two possible values for $t$ and each block mutually exclusive halves of the assignments blocked by $C$, the lemma holds.
    \end{proof}

    \begin{lemma}
        Given a clause, $C$, and another clause, $D$, such that all of the terms in $C$ also exist in $D$, then all of the assignments blocked by $D$ are also blocked by $C$.
    \end{lemma}
    \begin{proof}
        Given a clause, $C$, of a fixed yet arbitrary length:

        $C := [a, b, c, ...]$

        And another clause, $D$, containing all the terms in $C$ with possibly additional terms:

        $D := [a, b, c, ..., d, e, f, ...]$

        Want to show all the assignments blocked by $D$ are also blocked by $C$.

        We know that $C$ blocks all assignments that cause all the terms to evaluate to False.

        In other words, $C$ blocks all assignments where 
 
        $a = b = c = ... = False$

        Similarly, $D$ blocks all assignments where

        $a = b = c = ... = d = e = f = ... = False$

        Clearly all assignments consistent with the terminal assignments from $D$ are also consistent with the terminal assignments from $C$.

        Therefore, every assignment blocked by $D$ is also blocked by $C$.
    \end{proof}

    \begin{lemma}[Reduction]
        Given the following conditions:
        \begin{itemize}
            \item $A$ is a clause of length $k$
            \item $B$ is a clause of length $k$
            \item $A$ and $B$ share $k-1$ identical terms
            \item $A$ and $B$ share one terminal that is negated in one clause
            and positive in the other
            \item $C$ is a clause of length $k-1$ composed of the shared terms
            from $A$ and $B$
        \end{itemize}
        Then $A$ and $B$ imply $C$.
    \end{lemma}
    \begin{proof}
        Given clauses consistent with the description:

        $A := [a, b, c, ... i]$

        $B := [a, b, c, ..., -i]$

        where $a, b, c, ...$ is shared between the clauses and $i$ is a terminal not in $a, b, c, ...$

        Want to show this implies $C$.

        Recall by Lemma 5.3 that the same terminal cannot appear both negated and positive within the same clause and still block an assignment, so $i$ cannot exist in $a, b, c...$

        Consider the clause

        $C := [a, b, c, ...]$

        We know by Lemma 5.5 that if we select a terminal that's not in $C$, say $t$, then half of the assignments blocked by $C$ assign True to $t$ and the other half of the assignments blocked by $C$ assign False to $t$.

        Let this terminal $t$ that's not in $C$ be the terminal $i$ that's in $A$ and $B$.

        We know that $A$ blocks all assignments blocked by $C$ where $i$ is assigned the value of False.

        We know that $B$ blocks all assignments blocked by $C$ where $i$ is assigned the value of True.

        Since $A$ and $B$ both block mutually exclusive halves of the assignments blocked by $C$, then all of the assignments blocked by $C$ are blocked by $A$ and/or $B$ and we can say that $A$ and $B$ imply $C$.
    \end{proof}

    \begin{lemma}[Expansion]
        Given a clause, $C$, and a terminal, $t$, that's not in $C$, then two new clauses can be implied consisting of all the terms of $C$ appended to either the positive form of $t$ or the negated form of $t$.
    \end{lemma}
    \begin{proof}
        Given a clause, $C$, and a terminal, $t$, that's not in $C$:

        $C := [a, b, c, ...]$

        We can compose two new clauses:

        $D := [a, b, c, ..., t]$

        $E := [a, b, c, ..., -t]$

        Since all of the terms in $C$ exist in $D$ and $E$, then by Lemma 5.6 all of the assignments blocked by $D$ or $E$ are blocked by $C$ and we can say $D$ and $E$ are implied by $C$.   
    \end{proof}

    \begin{lemma}[General Lemma 5.7]
        If two clauses share the same terminal, $t$, such that $t$ is positive in one clause and negated in the other, then these clauses imply a new clause which is composed of all the terms in both clauses except terms containing $t$.
    \end{lemma}
    \begin{proof}
        Consider two clauses, 

        $C := [a, b, c, ..., t]$
        
        $D := [d, e, f, ..., -t]$

        Where within a clause, the same terminal does not repeat, but between clauses the same terminal may repeat.

        Want to show we can imply a clause consistent with the lemma description:

        $E := [a, b, c, ..., d, e, f, ...]$

        Let's define some additional clauses:

        $E' := [a, b, c, ..., d, e, f, ..., t]$

        $E'' := [a, b, c, ..., d, e, f, ..., -t]$

        By Lemma 5.6, we know that $C$ implies $E'$ because all the terms in $C$ exist in $E$.

        By Lemma 5.6, we know that $D$ implies $E''$.

        By Lemma 5.7, since $E'$ and $E''$ share all the same terms except for $t$, which is positive in one clause and negated in the other, we can create a new clause composed of all the shared terms in $E'$ and $E''$.

        Such a clause is already defined as $E$.

        Now there are a couple extra cases to consider:
        
        \begin{itemize}
            \item There is some overlap between $a, b, c, ...$ and $d, e, f, ...$
            \item There is the same terminal that's positive in $a, b, c, ...$ and negated in $d, e, f, ...$
        \end{itemize}

        First, if the same term exists in $a, b, c, ...$ and $d, e, f, ...$, then we can just remove one of the duplicates since one term being True implies an identical term being True.

        Secondly, if the same terminal exists, but is of the opposite form in $a, b, c, ...$ and $d, e, f, ...$
        then by Lemma 5.3, this clause will always be True and thus blocks no assignments. In this case, the lemma is vacuously true, but we disregard the clause as it is of no value.
    \end{proof}

    \begin{lemma}
        Given two clauses of lengths $k$ and $m$ that imply another clause by Lemma 5.9, the length of the implied clause will fall in the range $max(k, m)-1$ to $(k + m - 2)$ where the $max(k, m)$ represents the parameter with the greatest value.
    \end{lemma}
    \begin{proof}
        The smallest clause that can be implied by clauses of length $k$ and $m$ using Lemma 5.9 occur when all but one of the terms in one clause exist in the other.

        As such, the unique terms will come from the clause that's longer.

        Removing $t$, you are left with 1 less than the maximum of $k$ and $m$.

        The largest clause can be implied if there are no terms shared between the two clauses. In this case you subtract $1$ from the length of each clause to account for $t$ and since no duplicates will be removed, the resulting clause's length is $2$ less than the sum of the lengths of 
        the clauses.
    \end{proof}

    \begin{lemma}
        Given the following:
        \begin{itemize}
            \item A clause, $A$, of length less than $k$
            \item A clause, $B$, of length less than $k$
            \item A clause, $C$, of length less than $k$
            \item A clause, $D$, of length $k$ or $k - 1$
            \item A clause, $E$, of length $k$
            \item $A$ and $B$ imply $E$ by Lemma 5.9
            \item $C$ and $E$ imply $D$ by Lemma 5.9
        \end{itemize}
        Then $A$, $B$, and $C$, imply $D$ by processing only clauses with a maximum length of $k - 1$.
    \end{lemma}
    \begin{proof}
        
        Define the clauses in the following manner:

        A := [a, b, $\beta$, i]

        B := [c, d, $\delta$ -i]

        C := [-a, e, f, $\phi$]

        Then the following are derived by Lemma 5.9:
        
        $E = [a, b, \beta, c, d, \delta]$ (By $A$ and $B$)

        $D = [b, \beta, c, d, \delta, e, f, \phi]$ (By $C$ and $E$ or by $F$ and $B$)

        $F = [b, \beta, i, e, f, \phi]$ (By $A$ and $C$)

        Where $\beta, \delta, $ and $\phi$ are generic sets of terms in the instance.

        Consider the following figure:    
        
        \begin{figure}[h]
            \centering
            \includegraphics[scale=0.8]{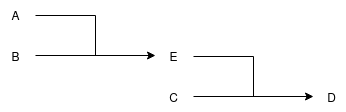}        
            \caption{A graph illustrating the derivations described by the lemma}
        \end{figure}

        Note that since $C$ and $E$ imply $D$, then $C$ and $E$ must share the same terminal such that it is positive in one clause and negated in the other.
        
        Since all of the terms in $E$ came from $A$ or $B$ (not including $i$), then such a term in $C$ must have the same term of the opposite form in $A$ or $B$.
        
        And since $A$ and $B$ are fixed yet arbitrary clauses treated in the same way, it does not matter which clause we pick as long as it is fixed for the rest of the proof. Let's pick a the terminal, $a$, from clause $A$. Now $C$ contains $-a$.
 
        Recall from Lemma 5.3 that if a clause blocks any assignments, it cannot contain the same term in both forms, so if $\beta, \delta$, or $\phi$ contain the same terminal in both forms then $D$ blocks no assignments and the resulting clause may be disregarded.

        We know $C$ is of length $k$ and there are two cases for $D$: (1) $D$ is of length $k$ or (2) $D$ is of length $k-1$.

        In the following equations, let the presence of a term represent a count of one, and the presence of a set of terms represent the number of terms in that set. If multiple sets of terms are shown in parentheses, let this represent the number of terms found in the intersection of both sets.

        Consider (1) the case where $D$ has length $k$ then we can define $k$ in terms of $D$:
        
        $k = b + c + d + e + f + \beta + \delta + \phi - (\beta \delta) - 
        (\beta \phi) - (\delta \phi) + (\beta \delta \phi)$

        length of $F = b + i + e + f + \beta + \phi - (\beta \phi)$

        Want to show length of $F$ is less than $k$

        $b + i + e + f + \beta + \phi - (\beta \phi) < 
        b + c + d + e + f + \beta + \delta + \phi - (\beta \delta) - 
        (\beta \phi) - (\delta \phi) + (\beta \delta \phi)$

        $\rightarrow$ $i + \beta + \phi - (\beta \phi) < 
        c + d + \beta + \delta + \phi - (\beta \delta) - 
        (\beta \phi) - (\delta \phi) + (\beta \delta \phi)$

        $\rightarrow$ $i < 
        c + d + \delta - (\beta \delta) - (\delta \phi) + (\beta \delta \phi)$

        As seen by using a Venn Diagram or other set intuition, $\delta - 
        (\beta \delta) - (\delta \phi) + (\beta \delta \phi)$, 
        represents the number of terminals in $\delta$ not in $\beta$
        and not in $\phi$.

        The lowest case for the right hand side of the inequality is where 
        this is 0, ie, all of the terms in $\delta$ are in $\beta$ or 
        $\phi$. In this case, the inequality becomes:

        $\rightarrow$ $i < c + d$

        Which is true as long as $c$ and $d$ exist in $B$.

        Want to show $c$ and $d$ always exists in $B$:

        Suppose not, then at most one of $c$ or $d$ exists.

        Recall we have the clauses:

        $A := [a, b, \beta, i]$
        
        $B := [c, d, \delta, -i]$
        
        $E := [a, b, \beta, c, d, \delta]$

        And since only $c$ or $d$ exist, we can redefine some clauses:

        $B := [x, -i]$
        
        $E := [a, b, \beta, x]$

        Where x is represents either $c$ or $d$, but not both.

        Note that no terms may exist in $\delta$ because any terms in $\delta$ could be extracted and treated as $c$ or $d$, but we know $x$ already represents the one of these terms that exist and the other term cannot exist.

        Notice the length of $A$ is the same as the length of $E$.

        This is a contradiction because the length of $A$ is given as less than $k$ and the length of $E$ is given as $k$.

        Therefore both $c$ and $d$ must exist and the inequality is true.

        Therefore $F$ is shorter than $k$ when $D$ is of length $k$.
    
        Consider (2) the case where the length of $D$ is $k - 1$:

        This means $k$ is one greater than the length of $D$, so $k$ is now:

        $k = b + c + d + e + f + \beta + \delta + \phi - (\beta \delta) - 
        (\beta \phi) - (\delta \phi) + (\beta \delta \phi) + 1$

        length of $F = b + i + e + f + \beta + \phi - (\beta \phi)$

        Want to show length of $F$ is less than $k$

        $b + i + e + f + \beta + \phi - (\beta \phi) < 
        b + c + d + e + f + \beta + \delta + \phi - (\beta \delta) - 
        (\beta \phi) - (\delta \phi) + (\beta \delta \phi) + 1$

        Similarly as before, the inequality will become:

        $\rightarrow$ $i < c + d + 1$

        Which is true as long as at least $c$ or $d$ exist.

        It was seen that both $c$ and $d$ must exist and since the proof does not rely on the length of $D$, the inequality is true.

        Therefore the length of $F$ is shorter than $k$ when the length of $D$ is $k - 1$.

        Since the length of $F$ is less than $k$ in all cases, you can derive $D$ by processing only clauses with a maximum length of $k-1$.
    \end{proof}

    \begin{lemma}
        Given the following:
        \begin{itemize}
            \item $A$ is a clause of length less than $k$
            \item $B$ is a clause of length $k$
            \item $C$ is a clause of length less than $k$
            \item $D$ is a clause of length $k$ or $k - 1$
            \item $A$ expands to imply $B$ by Lemma 5.8
            \item $B$ and $C$ imply $D$ by Lemma 5.9
        \end{itemize}
        Then $D$ can be implied by processing clauses of at most length $k - 1$.
    \end{lemma}
    \begin{proof}

        Let the clauses be defined as follows:

        $A := [a, b, \beta]$
        
        $B := [a, b, \beta, c, d, \delta]$

        $C_1$ := [-a, e, f, $\phi$]

        $C_2$ := [-c, e, f, $\phi$]

        $D_1$ := [b, $\beta$, c, d, $\delta$, e, f, $\phi$]

        $D_2$ := [a, b, $\beta$, d, $\delta$, e, f, $\phi$]

        E := [b, $\beta$, e, f, $\phi$]

        Where $\beta, \delta$, and $\phi$ are fixed, yet arbitrary sets of terms such that the clauses block at least one assignment.

        Consider the following figure

        \begin{figure}[h]
            \centering
            \includegraphics[scale=0.8]{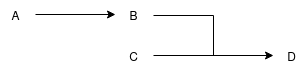}
            \caption{An illustration of the derivations described in the lemma}
        \end{figure}

        Notice that $C$ has to contain a term from $B$ in the opposite form by Lemma 5.9.

        Notice that $B$ is made up of terms from $A$ and terms not in $A$
        by Lemma 5.6.

        The term in $C$ which is opposite from the term in $B$ can therefore
        be opposite (1) from a term in $A$ (in this case, use $C_1$ and $D_1$) or (2) a term not in $A$ (in this case use $C_2$ and $D_2$).
        
        (1) Consider the opposite form term in $C$ is in $A$ (use $C_1$ and $D_1$):

        Recall the clauses:

        $A := [a, b, \beta]$
        
        $C_1$ := [-a, e, f, $\phi$]

        $D_1$ := [b, $\beta$, c, d, $\delta$, e, f, $\phi$]

        $E := [b, \beta, e, f, \phi]$

        Notice $A$ and $C_1$ share an opposite term, so they can derive a clause, $E$, by Lemma 5.9. 
        
        Notice all the terms in $E$ exist in $D_1$ so $E$ can be expanded to derive $D_1$ by Lemma 5.8.

        Consider the path of implication in the following figure:

        \begin{figure}[h]
            \centering
            \includegraphics[scale=0.8]{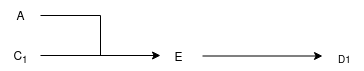}
            \caption{An illustration of the derivations by clauses $A$, $C_1$, and $E$}
        \end{figure}

        Want to show you only have to process clauses whose length is
        less than $k$ to derive $D_1$.

        Since $D_1$ is derived using only $A$, $C_1$, and $E$ and $A$ and $C$ are given to be shorter than $k$, want to show length of $E$ is less than $k$.

        Consider two cases: (1a) $D_1$ is of length $k$ and (1b) $D_1$ is of length $k - 1$
        
        Consider (1a) where $D_1$ is of length $k$:

        In the following equations, let the presence of a term represent a count of one, and the presence of a set of terms represent the number of terms in that set. If multiple sets of terms are shown in parentheses, let this represent the number of terms found in the intersection of both sets.

        Since $D_1$ is of length $k$, we can define $k$ as follows:

        $k = b + c + d + e + f + \beta + \delta + \phi - (\beta \delta) 
        - (\beta \phi) - (\delta \phi) + (\beta \delta \phi)$

        Length of $E = b + e + f + \beta + \phi - (\beta \phi)$
        
        Want to show the length of $E$ is less than $k$:

        $b + e + f + \beta + \phi - (\beta \phi) < b + c + d + e + f + 
        \beta + \delta + \phi - (\beta \delta) - (\beta \phi) - (\delta \phi) + (\beta \delta \phi)$

        $\rightarrow 0 < c + d + \delta - (\beta \delta) 
        - (\delta \phi) + (\beta \delta \phi)$

        Which is true as long as at least one term exists on the R.H.S.

        Want to show at least one term exists on the R.H.S.

        Suppose not, then no terms exist on the R.H.S. and we can redefine some clauses.

        Recall the original clause definitions:

        $A := [a, b, \beta]$

        $B := [a, b, \beta, c, d, \delta]$

        Since no terms exist on the R.H.S. we can redefine some clauses:

        $B := [a, b, \beta]$

        Note that no terms may exist in $\delta$ because any term in $\delta$ can be extracted and treated as $c$ or $d$ and those are explicitly removed from existence.

        Notice $A$ is exactly $B$.

        This is a contradiction because the length of $A$ is given as less than $k$ while the length of $B$ is given as $k$.

        Therefore at least one term must exist on the R.H.S. and the inequality holds.

        Therefore $E$ is shorter than $k$ when $D$ is of length $k$ and we use $C_1$ and $D_1$.

        Consider (1b) $D_1$ is of length $k - 1$:

        Then the length of $k$ is redefined as:

        $k = b + c + d + e + f + \beta + \delta + \phi - (\beta \delta) 
        - (\beta \phi) - (\delta \phi) + (\beta \delta \phi) + 1$

        Similarly as before, the inequality becomes:

        $\rightarrow 0 < c + d + 1$

        Which is always true so the length of $E$ is indeed less than $k$ when $D$ is of length $k - 1$ and we use $C_1$ and $D_1$.

        (2) Consider the case using $C_2$ and $D_2$:
        
        Recall the clauses:

        $A := [a, b, \beta]$

        $C_2 = [-c, e, f, \phi]$

        $D_2 = [a, b, \beta, d, \delta, e, f, \phi]$

        Since all of the terms in $A$ exist in $D_2$, we can expand $A$ to $D_2$ using Lemma 5.8

        Since the length of $A$ is given as less than $k$, we can derive $D$ by processing clauses with a maximum length of $k - 1$.
        
        Since the possible lengths of $D$ are $k$ or $k - 1$ and in both cases we derive $D$ by processing clauses with a maximum length of $k - 1$ then $D$ can always be derived by processing clauses with a maximum length of $k - 1$.
    \end{proof}

    \begin{lemma}
        Given an instance of 3SAT, you can expand all of the given clauses to the point where you are considering clauses of length $n$.
    \end{lemma}
    \begin{proof}
        Given an instance of 3SAT, we know all clauses are of length 3.

        If we want to consider a generic $n$-terminal clause, $B$, that's implied by a given clause, $A$, then by Lemma 5.6 we know it's implied if all the terms in $A$ exist in $B$.
    \end{proof}

    \begin{lemma}
        If you expand given 3-t clauses as described in Lemma 5.13, you will
        derive $2^n$ unique clauses of length $n$ iff the instance is unsatisfiable.
    \end{lemma}
    \begin{proof}
        Want to show an unsatisfiable instance $\implies$ $2^n$ unique n-terminal clauses can be derived from the given 3-t clauses:

        By lemma 5.4, a clause of length $n$ blocks 1 assignment. 

        Recall an instance is unsatisfiable iff all $2^n$ assignments are
        blocked.

        If a 3-terminal clause blocks an assignment, then it also implies the corresponding n-terminal assignment because there is one possible n-terminal clause for any given assignment.

        Notice that for each of these n-terminal clauses, they must contain three terms from at least one given clause. If they didn't, then the assignment blocked by that n-terminal clause would not be blocked and the instance would be satisfiable.

        Since each n-terminal clause blocks one assignment, blocking all assignments requires $2^n$ n-terminal clauses.

        Want to show $2^n$ unique n-terminal clauses are derived by the given 3-t clauses $\implies$ then the instance is unsatisfiable.

        By lemma 5.4, a clause of length $n$ blocks 1 assignment. 

        Therefore if there are $2^n$ unique n-terminal clauses, then
        all $2^n$ assignments will be blocked.

        Note that there will be no overlap because each n-terminal clause
        sets the value for each terminal and overlap would imply the same
        terminal having two values by the same assignment which is impossible.
    \end{proof}

    \begin{lemma}
        The n-terminal clauses described in Lemma 5.14 can be reduced to derive any pair of contradicting 1-terminal clauses.
    \end{lemma}
    \begin{proof}
        Given $2^n$ n-terminal clauses, want to show you can imply any pair of contradicting 1-terminal clauses by lemma 5.7.

        Algorithm: 

        Pick a terminal that will not exist in the final 1-terminal clauses.

        Notice half of the existing n-terminal clauses have that terminal assigned the value of False and the other half have that terminal assigned the value of True.

        Pick one clause that blocks an assignment where the terminal is True.

        Then there exists an assignment for each possible value for the remaining n-1 terminals.

        Therefore, there must exist another clause that shares all of the same terms, but where that one terminal is assigned the value of False.

        Using these two clauses, we can create a new clause by lemma 5.7.

        Now all of the clauses of length n - 1 do not contain that terminal.

        Repeat this process while never selecting the same terminal twice
        until you are left with two contradicting 1-terminal clauses.

        Each clause is guaranteed to have a matching clause since every possible combination of terminal assignments exists and when you use Lemma 5.7 to create new clauses, the rest of all of the clauses remain the same except the selected terminal is removed.
    \end{proof}

    \begin{lemma}
        Contradicting 1-terminal clauses can be expanded to imply every possible clause.
    \end{lemma}
    \begin{proof}
        Let the following clauses be a pair of contradicting 1-t clauses:

        $[a]$

        $[-a]$

        Any possible clause could either contain $a$, $-a$, or neither.

        By lemma 5.8, we can expand to any clause that contains $a$ or $-a$.

        Now want to show these clauses imply another clause that does not contain  $a$ or $-a$.

        Let the following be a generic k-terminal clause that does not contain $a$ or $-a$:

        $[b, c, d, ...]$

        By Lemma 5.8, we know the 1-terminal clauses imply the following clauses:

        $[a, b]$

        $[-a, c, d, ...]$

        By Lemma 5.9, these clauses imply:

        $[b, c, d, ...]$

        Therefore a set of contradicting 1-terminal clauses can imply any clause containing $a, -a$, or neither, which encompasses every possible clause.
    \end{proof}

    \begin{lemma}
        Given the following:
        \begin{itemize}
            \item A, B, C, and D, are clauses shorter than k
            \item E and F are clauses of length k
            \item G is a clause of length k or k - 1
            \item A and B imply E by Lemma 5.9
            \item C and D imply F by Lemma 5.9
            \item E and F imply G by Lemma 5.9
        \end{itemize}
        Then G can be implied by processing clauses with a maximum length of k - 1.
    \end{lemma}
    \begin{proof}
        Consider the following figure:

        \begin{figure}[h]
            \centering
            \includegraphics[scale=0.8]{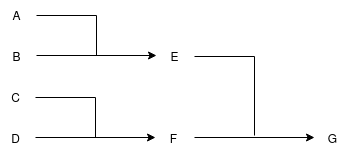}
            \caption{An illustration of the derivations described in the lemma}
        \end{figure}

        Let the clauses be defined as follows:

        $A := [a, b, \beta, i]$

        $B := [c, d, \delta, -i]$
        
        $C := [-a, f, \phi, j]$
        
        $D := [g, h, \gamma, -j]$

        Then the following are implied by Lemma 5.9:

        $E = [a, b, \beta, c, d, \delta]$
        
        $F = [-a, f, \phi, g, h, \gamma]$
        
        $G = [b, \beta, c, d, \delta, f, \phi, g, h, \gamma]$
        
        Where $\beta, \delta, \phi,$ and $\gamma$ are fixed yet arbitrary sets of terms such that the clauses block at least one assignment.

        Notice that $E$ and $F$ have to share a term of the opposite form in order to imply $G$ by Lemma 5.9. 
        
        All of the terms in $E$ and $F$ came from $A, B, C,$ and $D$. 
        
        Since $A, B, C,$ and $D$ are all arbitrary clauses, selecting which term to negate does not have an effect on the outcome as long as one form of the term exists in $E$ and the other form exists in $F$.
        
        Here the opposite term is shared between $A$ and $D$, but any term that appears positive in $A$ or $B$ and negated in $C$ or $D$ or vice versa will yield the same results.
 
        Want to show $G$ can be derived by processing clauses with a maximum length of $k - 1$.

        Define additional implications by Lemma 5.9:

        $H = [b, \beta, i, f, \phi, j]$ (By clauses $A$ and $C$)

        $I = [c, d, \delta, b, \beta, f, \phi, j]$ (By clauses $B$ and $H$)

        $J = [g, h, \gamma, c, d, \delta, b, \beta, f, \phi]$ (By clauses $D$ and $I$)

        Notice $J$ is equivalent to $G$.

        Want to show all clauses used to derive $J$ are shorter than $k$.

        Want to show $A$, $C$, $B$, $H$, $D$, and $I$ are shorter than $k$.

        It was given that $A$, $B$, $C$, and $D$ are shorter than $k$ so just want to show $H$ and $I$ are shorter than $k$.

        Want to show (1) $H$ is shorter than $k$ and (2) $I$ is shorter than $k$

        (1) Want to show $H$ is shorter than $k$

        Two cases to consider (1a) $G$ is of length $k$ and (1b) $G$ is of length $k - 1$

        (1a) $G$ is of length $k$

        In the following equations, let the presence of a term represent a count of one, and the presence of a set of terms represent the number of terms in that set. If multiple sets of terms are shown in parentheses, let this represent the number of terms found in the intersection of both sets.

        Since $G$ is of length $k$ we can define $k$ as follows:

        $k = b + c + d + f + g + h
            + \beta + \delta + \phi + \gamma
            - (\beta \delta) - (\beta \phi) - (\beta \gamma) - (\delta \phi) - (\delta \gamma) -(\phi \gamma)
            + (\beta \delta \phi) + (\beta \delta \gamma) + (\beta \phi \gamma) + (\delta \phi \gamma)
            - (\beta \delta \phi \gamma)
        $

        Length of $H = b + i + f + j + \beta + \phi - (\beta \phi)$
        
        Want to show the length of $H$ is less than $k$:

        $b + i + f + j + \beta + \phi - (\beta \phi)$
        $<$
        $b + c + d + f + g + h
            + \beta + \delta + \phi + \gamma
            - (\beta \delta) - (\beta \phi) - (\beta \gamma) - (\delta \phi) - (\delta \gamma) -(\phi \gamma)
            + (\beta \delta \phi) + (\beta \delta \gamma) + (\beta \phi \gamma) + (\delta \phi \gamma)
            - (\beta \delta \phi \gamma)
        $

        $\rightarrow$
        $i + j$
        $<$
        $c + d + g + h
            + \delta + \gamma
            - (\beta \delta) - (\beta \gamma) - (\delta \phi) - (\delta \gamma) -(\phi \gamma)
            + (\beta \delta \phi) + (\beta \delta \gamma) + (\beta \phi \gamma) + (\delta \phi \gamma)
            - (\beta \delta \phi \gamma)
        $

        Which is true as long as at least three terms exist on the R.H.S.

        Consider the following cases:

        \begin{itemize}
            \item No terms exist on the R.H.S.
            \item Exactly one term exists on the R.H.S.
            \item Exactly two terms exist on the R.H.S.
        \end{itemize}

        Consider case 1 where no terms exist on the R.H.S.

        Recall the clauses:

        $A := [a, b, \beta, i]$

        $B := [c, d, \delta, -i]$

        $E := [a, b, \beta, c, d, \delta]$

        But since no terms exist on the R.H.S., specifically $c$ and $d$ don't exist, we can redefine the clauses:

        $B := [-i]$

        $E := [a, b, \beta]$

        Notice that no terms may exist in $\delta$ because any term in $\delta$ can be extracted and treated as $c$ or $d$, but we defined these to not exist.

        Notice the length of $E$ is shorter than the length of $A$.

        This is a contradiction because the length of $A$ is given as less than $k$ while the length of $E$ is given as $k$.

        Therefore this case is impossible and at least one term must exist on the R.H.S.
 
        Consider case 2 where exactly one term exists on the R.H.S.

        Recall the clauses:

        $A := [a, b, \beta, i]$

        $B := [c, d, \delta, -i]$

        $E := [a, b, \beta, c, d, \delta]$

        Recall at least one term must exist on the R.H.S., specifically either $c$ or $d$ must exist, but not both.

        We know $c$ or $d$ must exist because if neither exist, there's a contradiction (see case 1).

        Since these terms are treated in the same way, let's pick $c$ to be the term that exists.

        We can redefine some clauses as follows:

        $B := [c, -i]$

        $E := [a, b, \beta, c]$

        Notice the length of $E$ is the same as the length of $A$.

        This is a contradiction because the length of $A$ is given as less than $k$ while the length of $E$ is given as $k$.

        Therefore at least two terms must exist on the R.H.S.

        Consider case 3 where exactly two terms exists on the R.H.S.

        Notice how if both $c$ and $d$ do not exist, then a contradiction is reached. 

        Therefore both $c$ and $d$ must exist and the rest of the terms on the R.H.S. may not exist.

        Recall we have the clauses:

        $C := [-a, f, \phi, j]$

        $D := [g, h, \gamma, -j]$

        $F := [-a, f, \phi]$

        But since no terms on the R.H.S. exist besides $c$ and $d$, then $g$, $h$, and $\gamma$ may not exist. 

        Notice that no terms in $\gamma$ may exist because if any terms in $\gamma$ exist, then they can be extracted and treated as $g$ or $h$.

        Then we can redefine the clauses:

        $D := [-j]$

        $F := [-a, f, \phi]$

        Notice that $F$ is shorter than $C$.

        This is a contradiction because the length of $C$ is given as less than $k$ while the length of $F$ is given as $k$.

        Therefore at least three terms must exist on the R.H.S.

        Since at least three terms must exist on the R.H.S., the inequality is true.

        Therefore $H$ is shorter than $k$ when the length of $G$ is $k$.

        (1b) $G$ is of length $k - 1$

        If $G$ is of length $k - 1$, then k is defined as:

        $k = b + c + d + f + g + h
            + \beta + \delta + \phi + \gamma
            - (\beta \delta) - (\beta \phi) - (\beta \gamma) - (\delta \phi) - (\delta \gamma) -(\phi \gamma)
            + (\beta \delta \phi) + (\beta \delta \gamma) + (\beta \phi \gamma) + (\delta \phi \gamma)
            - (\beta \delta \phi \gamma)
            + 1
        $

        Length of $H = b + i + f + j + \beta + \phi - (\beta \phi)$

        Similarly to before, the inequality becomes:

        $\rightarrow$
        $i + j$
        $<$
        $c + d + g + h
            + \delta + \gamma
            - (\beta \delta) - (\beta \gamma) - (\delta \phi) - (\delta \gamma) -(\phi \gamma)
            + (\beta \delta \phi) + (\beta \delta \gamma) + (\beta \phi \gamma) + (\delta \phi \gamma)
            - (\beta \delta \phi \gamma)
            + 1
            $

        Which is true as long as at least two terms exist on the R.H.S.

        It was already shown that at least three terms must exist on the R.H.S. and the proof does not rely on the length of $G$ so the inequality holds.

        Therefore $H$ is always shorter than $k$. 

        (2) Want to show $I$ is shorter than $k$

        Two cases to consider (2a) $G$ is of length $k$ and (2b) $G$ is of length $k - 1$

        (2a) $G$ is of length $k$

        Since $G$ is of length $k$ we can define $k$ as follows:

        $k = b + c + d + f + g + h
            + \beta + \delta + \phi + \gamma
            - (\beta \delta) - (\beta \phi) - (\beta \gamma) - (\delta \phi) - (\delta \gamma) -(\phi \gamma)
            + (\beta \delta \phi) + (\beta \delta \gamma) + (\beta \phi \gamma) + (\delta \phi \gamma)
            - (\beta \delta \phi \gamma)
        $

        Length of $I = c + d + b + f + j
        + \delta + \beta + \phi
        - (\delta \beta) - (\delta \phi) - (\beta \phi)
        + (\delta \beta \phi)
        $

        Want to show the length of $I$ is less than $k$:

        $c + d + b + f + j
        + \delta + \beta + \phi
        - (\delta \beta) - (\delta \phi) - (\beta \phi)
        + (\delta \beta \phi)
        $
        $<$
        $b + c + d + f + g + h
            + \beta + \delta + \phi + \gamma
            - (\beta \delta) - (\beta \phi) - (\beta \gamma) - (\delta \phi) - (\delta \gamma) -(\phi \gamma)
            + (\beta \delta \phi) + (\beta \delta \gamma) + (\beta \phi \gamma) + (\delta \phi \gamma)
            - (\beta \delta \phi \gamma)
        $

        $\rightarrow$
        $j
        $
        $<$
        $g + h
            + \gamma
            - (\beta \gamma) - (\delta \gamma) -(\phi \gamma)
            + (\beta \delta \gamma) + (\beta \phi \gamma) + (\delta \phi \gamma)
            - (\beta \delta \phi \gamma)
        $

        Which is true as long as at least two terms exist on the R.H.S.

        Suppose not, then a maximum of one term exists on the R.H.S.

        Recall the clauses

        $C := [-a, f, \phi, j]$

        $D := [g, h, \gamma, -j]$

        $F := [-a, f, \phi, g, h, \gamma]$

        Due to the restriction of a maximum of one term existing on the R.H.S., we can redefine some clauses:

        $D := [x, -j]$

        $F := [-a, f, \phi, x]$

        Where $x$ represents a maximum of one term.

        Note that if $x$ was more than one term, then the two terms could be used as $g$ and $h$, but we know these do not exist.

        Notice the length of $F$ is at most the length of $C$.

        This is a contradiction because the length of $C$ is given as less than $k$ while the length of $F$ is given as $k$

        Therefore at least two terms exist on the R.H.S. and the inequality is true.

        Therefore I is shorter than $k$ when $G$ is of length $k$.

        (2b) $G$ is of length $k - 1$

        Similarly as before, we define k in terms of the length of G:

        $k = b + c + d + f + g + h
            + \beta + \delta + \phi + \gamma
            - (\beta \delta) - (\beta \phi) - (\beta \gamma) - (\delta \phi) - (\delta \gamma) -(\phi \gamma)
            + (\beta \delta \phi) + (\beta \delta \gamma) + (\beta \phi \gamma) + (\delta \phi \gamma)
            - (\beta \delta \phi \gamma)
            + 1
        $

        Length of $I = c + d + b + f + j
        + \delta + \beta + \phi
        - (\delta \beta) - (\delta \phi) - (\beta \phi)
        + (\delta \beta \phi)
        $

        Similarly to before, the inequality becomes
        
        $\rightarrow$
        $j$
        $<$
        $g + h
            + \gamma
            - (\beta \gamma) - (\delta \gamma) -(\phi \gamma)
            + (\beta \delta \gamma) + (\beta \phi \gamma) + (\delta \phi \gamma)
            - (\beta \delta \phi \gamma)
            + 1
        $

        Which is true as long as at least one term exists on the R.H.S.

        It was already shown at least two terms exist on the R.H.S. and the proof does not rely on the length of $G$ so the inequality is true.

        Therefore, the length of $I$ is shorter than $k$ when $G$ is of length $k - 1$.

        Since $G$ can be derived by processing only $A$, $B$, $C$, $D$, $H$, and $I$ and all of those clauses are shorter than $k$, we can derive $D$ by processing clauses with a maximum length of $k - 1$.
    \end{proof}

    \begin{lemma}
        Given the following:
        \begin{itemize}
            \item a clause, $A$, of length less than $k$
            \item a clause, $B$, of length less than $k$
            \item a clause, $C$, of length less than $k$
            \item a clause, $D$, of length $k$
            \item a clause $E$, of length $k$
            \item a clause $F$, of length $k - 1$ or $k$
            \item A and B imply D by Lemma 5.9
            \item C expands to E by Lemma 5.8
            \item D and E imply F by Lemma 5.9
        \end{itemize}
        Then $F$ can be implied by only processing clauses with a maximum length of $k-1$.
    \end{lemma}
    \begin{proof}
        Consider the following figure:

        \begin{figure}[h]
            \centering
            \includegraphics[scale=0.8]{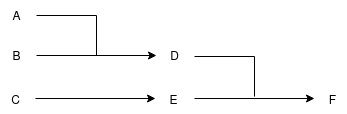}
            \caption{An illustration of the derivations described by the lemma}
        \end{figure}

        Let the following clauses be defined:

        $A := [a, b, \beta, i]$

        $B := [c, d, \delta, -i]$

        $C_1 := [-a, f, \phi]$
        
        $C_2 := [e, f, \phi]$
        
        Then the following clauses are implied:
        
        $D = [a, b, \beta, c, d, \delta]$
        
        $E_1 = [-a, f, \phi, g, h, \gamma]$
        
        $E_2 = [e, f, \phi, -a, h, \gamma]$
        
        $F_1 = [b, \beta, c, d, \delta, f, \phi, g, h, \gamma]$
        
        $F_2 = [b, \beta, c, d, \delta, e, f, \phi, h, \gamma]$
        
        Where $\beta, \delta, \phi$ and $\gamma$ are generic sets of terms
        such that A, B, and C block at least one clause by Lemma 5.3.
        
        Notice E and D must share a term of the opposite form and there are two possible cases for what this term is: 
            \begin{itemize}
                \item it exists in C (use $C_1$, $E_1$, and $F_2$)
                \item it does not exist in C (use $C_2$, $E_2$, and $F_2$)
            \end{itemize}
        Either way the opposite term must exist in A or B since it must exist in D and D is composed of terms from A or B.
        
        Since A and B are both generic clauses treated in the same way, it does not matter which clause the term exists in as long as it is fixed.
        
        Let's pick $A$ to contain the opposite form term from $C$.
        
        Consider (1) the opposite form term exists in C (use $C_1$, $E_1$, and $F_1$).

        Want to show $F_1$ can be derived by processing clauses with a maximum length of k - 1.

        Let the following clauses be defined:

        $G = [b, \beta, i, f, \phi]$ (By clause $A$ and $C_1$ with Lemma 5.9)

        $H = [b, \beta, f, \phi, c, d, \delta]$ (By clause $G$ and $B$ with Lemma 5.9)

        Since all of the terms in $H$ exist in $F_1$, Lemma 5.8 can be used to derive $F_1$ from $H$.

        Want to show (1a) $G$ is shorter than $k$ and (1b) $H$ is shorter than $k$

        (1a) $G$ is shorter than $k$

        Consider two cases (1ai) $F_1$ is of length $k$ and (1aii) $F_1$ is of length $k-1$.

        (1ai) $F_1$ is of length $k$

        In the following equations, let the presence of a term represent a count of one, and the presence of a set of terms represent the number of terms in that set. If multiple sets of terms are shown in parentheses, let this represent the number of terms found in the intersection of both sets.
        
        Recall $F_1$ is of length $k$ so $k$ can be defined as follows:
        
        $k = b + c + d + f + g + h 
        + \beta + \delta + \phi + \gamma
        - (\beta \delta) - (\beta \phi) - (\beta \gamma) - (\delta \phi) - (\delta \gamma) - (\phi \gamma)
        + (\beta \delta \phi) + (\beta \delta \gamma) + (\beta \phi \gamma) + (\delta \phi \gamma)
        - (\beta \delta \phi \gamma)
        $

        Length of $G = b + i + f + \beta + \phi - (\beta \phi)$

        Want to show length of $G$ is less than $k$:

        $b + i + f + \beta + \phi - (\beta \phi)$
        $<$
        $b + c + d + f + g + h 
        + \beta + \delta + \phi + \gamma
        - (\beta \delta) - (\beta \phi) - (\beta \gamma) - (\delta \phi) - (\delta \gamma) - (\phi \gamma)
        + (\beta \delta \phi) + (\beta \delta \gamma) + (\beta \phi \gamma) + (\delta \phi \gamma)
        - (\beta \delta \phi \gamma) 
        $

        $\rightarrow$ $i$
        $<$
        $c + d + g + h 
         + \delta + \gamma
        - (\beta \delta) - (\beta \gamma) - (\delta \phi) - (\delta \gamma) - (\phi \gamma)
        + (\beta \delta \phi) + (\beta \delta \gamma) + (\beta \phi \gamma) + (\delta \phi \gamma)
        - (\beta \delta \phi \gamma) 
        $

        Which is true as long as at least two terms exist on the R.H.S.

        Want to show at least two terms exist on the R.H.S.

        Suppose not, then at most one term exists on the R.H.S.

        Recall the clauses

        $A := [a, b, \beta, i]$

        $D := [a, b, \beta, c, d, \delta]$

        We can redefine some clauses since at most one term may exists on the R.H.S.

        $D := [a, b, \beta, x]$

        Where $x$ represents one term from $c$, $d$, or $\delta$.

        Note that if $x$ was more than one term, the two terms could be extracted and treated as $c$ and $d$, but we know at most one of these terms exist.
        
        Note the length of $D$ is the same as the length of $A$.

        This is a contradiction because the length of $A$ is given as less than $k$ while the length of $D$ is given as $k$.

        Therefore at least two terms exist on the R.H.S. and the inequality is true.

        Therefore $G$ is shorter than $k$ when using $F_1$ and the length of $F_1$ is $k$.

        (1aii) $F_1$ is of length $k - 1$

        Recall $F_1$ is of length $k - 1$ so $k$ can be defined as follows:
        
        $k = b + c + d + f + g + h 
        + \beta + \delta + \phi + \gamma
        - (\beta \delta) - (\beta \phi) - (\beta \gamma) - (\delta \phi) - (\delta \gamma) - (\phi \gamma)
        + (\beta \delta \phi) + (\beta \delta \gamma) + (\beta \phi \gamma) + (\delta \phi \gamma)
        - (\beta \delta \phi \gamma)
        + 1
        $

        Length of $G = b + i + f + \beta + \phi - (\beta \phi)$

        Want to show length of $G$ is less than $k$:

        $b + i + f + \beta + \phi - (\beta \phi)$
        $<$
        $b + c + d + f + g + h 
        + \beta + \delta + \phi + \gamma
        - (\beta \delta) - (\beta \phi) - (\beta \gamma) - (\delta \phi) - (\delta \gamma) - (\phi \gamma)
        + (\beta \delta \phi) + (\beta \delta \gamma) + (\beta \phi \gamma) + (\delta \phi \gamma)
        - (\beta \delta \phi \gamma) 
        + 1
        $

        Similarly as before, the inequality will become:

        $\rightarrow$ $i$
        $<$
        $c + d + g + h 
         + \delta + \gamma
        - (\beta \delta) - (\beta \gamma) - (\delta \phi) - (\delta \gamma) - (\phi \gamma)
        + (\beta \delta \phi) + (\beta \delta \gamma) + (\beta \phi \gamma) + (\delta \phi \gamma)
        - (\beta \delta \phi \gamma) 
        + 1
        $

        Which is true as long as at least one term exists on the R.H.S.

        We already showed at least two terms exist on the R.H.S. and the proof does not rely on the length of $F_1$.

        Therefore at least one term exists on the R.H.S. and the inequality is true.

        Therefore the length of $G$ is less than $k$ when using $F_1$ and the length of $F_1$ is $k - 1$.

        (1b) Want to show $H$ is shorter than $k$.

        We have two cases to consider: (1bi) $F_1$ is of length $k$ and (1bii) $F_1$ is of length $k-1$

        Consider (1bi) $F_1$ is of length $k$

        Recall $F_1$ is of length $k$, so we can define $k$ as follows:

        $k = b + c + d + f + g + h 
        + \beta + \delta + \phi + \gamma
        - (\beta \delta) - (\beta \phi) - (\beta \gamma) - (\delta \phi) - (\delta \gamma) - (\phi \gamma)
        + (\beta \delta \phi) + (\beta \delta \gamma) + (\beta \phi \gamma) + (\delta \phi \gamma)
        - (\beta \delta \phi \gamma)
        $

        Length of $H = b + f + c + d 
        + \beta + \delta + \phi
        - (\beta \delta) - (\beta \phi) - (\delta \phi)
        + (\beta \delta \phi) $

        Want to show the length of $H$ is less than $k$.

        $b + f + c + d 
        + \beta + \delta + \phi
        - (\beta \delta) - (\beta \phi) - (\delta \phi)
        + (\beta \delta \phi) $
        $<$
        $b + c + d + f + g + h 
        + \beta + \delta + \phi + \gamma
        - (\beta \delta) - (\beta \phi) - (\beta \gamma) - (\delta \phi) - (\delta \gamma) - (\phi \gamma)
        + (\beta \delta \phi) + (\beta \delta \gamma) + (\beta \phi \gamma) + (\delta \phi \gamma)
        - (\beta \delta \phi \gamma)
        $

        $\rightarrow$
        $b + f + c + d$
        $<$
        $b + c + d + f + g + h 
        + \gamma
        - (\beta \gamma) - (\delta \gamma) - (\phi \gamma)
        + (\beta \delta \gamma) + (\beta \phi \gamma) + (\delta \phi \gamma)
        - (\beta \delta \phi \gamma)
        $

        $\rightarrow$
        $0$
        $<$
        $g + h 
        + \gamma
        - (\beta \gamma) - (\delta \gamma) - (\phi \gamma)
        + (\beta \delta \gamma) + (\beta \phi \gamma) + (\delta \phi \gamma)
        - (\beta \delta \phi \gamma)
        $

        Which is true if at least one term exists on the R.H.S.

        Want to show at least one term exists on the R.H.S.

        Suppose not, then no terms exist on the R.H.S.

        Recall the clauses

        $C_1 := [-a, f, \phi]$

        $E_1 := [-a, f, \phi, g, h, \gamma]$

        Since no terms exist on the R.H.S., we can redefine some clauses:

        $E_1 := [-a, f, \phi]$

        Note that no other terms may exist in $E_1$ because any other terms could be used as $g$ or $h$, but we know these don't exist.

        Notice $E_1$ is exactly $C_1$.

        This is a contradiction because the length of $C$ is given as less than $k$ while the length of $E$ is given as $k$.

        Therefore at least one term exists on the R.H.S. and the inequality is true.

        Therefore $H$ is shorter than $k$ when using $F_1$ and the length of $F_1$ is $k$.

        Consider (1bii) $F_1$ is of length $k-1$

        Since $F_1$ is of length $k - 1$, we can define $k$ as follows:

        $k = b + c + d + f + g + h 
        + \beta + \delta + \phi + \gamma
        - (\beta \delta) - (\beta \phi) - (\beta \gamma) - (\delta \phi) - (\delta \gamma) - (\phi \gamma)
        + (\beta \delta \phi) + (\beta \delta \gamma) + (\beta \phi \gamma) + (\delta \phi \gamma)
        - (\beta \delta \phi \gamma)
        + 1
        $

        Length of $H = b + f + c + d 
        + \beta + \delta + \phi
        - (\beta \delta) - (\beta \phi) - (\delta \phi)
        + (\beta \delta \phi) $

        Want to show the length of H is less than k.

        $b + f + c + d 
        + \beta + \delta + \phi
        - (\beta \delta) - (\beta \phi) - (\delta \phi)
        + (\beta \delta \phi) $
        $<$
        $b + c + d + f + g + h 
        + \beta + \delta + \phi + \gamma
        - (\beta \delta) - (\beta \phi) - (\beta \gamma) - (\delta \phi) - (\delta \gamma) - (\phi \gamma)
        + (\beta \delta \phi) + (\beta \delta \gamma) + (\beta \phi \gamma) + (\delta \phi \gamma)
        - (\beta \delta \phi \gamma)
        + 1
        $

        Similarly to before, this will derive 

        $\rightarrow$
        $0$
        $<$
        $g + h 
        + \gamma
        - (\beta \gamma) - (\delta \gamma) - (\phi \gamma)
        + (\beta \delta \gamma) + (\beta \phi \gamma) + (\delta \phi \gamma)
        - (\beta \delta \phi \gamma)
        + 1
        $

        Which is always true.
        
        Therefore the length of $H$ is less than $k$ when using $F_1$ and the length of $F_1$ is $k - 1$.

        Consider (2) the opposite form term does not exist in $C$, use $C_2$, $E_2$, and $F_2$.

        Recall we have the clauses:

        $C_2 := [e, f, \phi]$

        $F_2 := [b, \beta, c, d, \delta, e, f, \phi, h, \gamma]$

        Notice all of the terms in $C_2$ exist in $F_2$ so $C_2$ can expand to $F_2$ using Lemma 5.9.

        Since the length of $C$ is given as less than $k$, we can derive $F_2$ by processing clauses with a maximum length of $k - 1$.

        Since $F_1$ and $F_2$ are all possible cases of $F$ and they can be derived without processing clauses longer than length $k-1$, $F$ can be derived by processing only clauses with a maximum length of $k-1$.
    \end{proof}

    \begin{lemma}
        Given the following:
        \begin{itemize}
            \item a clause $A$, of length less than k
            \item a clause $B$, of length less than k
            \item a clause $C$, of length k
            \item a clause $D$, of length k
            \item a clause $E$, of length k or k - 1
            \item $A$ expands to $C$ by lemma 5.8
            \item $B$ expands to $D$ by lemma 5.8
            \item $C$ and $D$ imply $E$ by lemma 5.7
        \end{itemize}
        Then $E$ can be derived by processing clauses whose length is at most k - 1.
    \end{lemma}
    \begin{proof}
        Consider the following figure:

        \begin{figure}[h]
            \centering
            \includegraphics[scale=0.8]{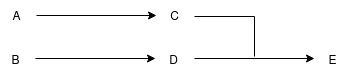}
            \caption{An illustration of the derivations described in this lemma}
        \end{figure}

        In the following clause definitions let $\beta, \delta, \phi,$ and $\gamma$ be fixed yet arbitrary set of terms in which the clauses block at least one assignment by Lemma 5.3.

        In order for $C$ and $D$ to imply E by Lemma 5.9, they have to share a terminal which is positive in one clause and negated in the other.

        Note that $C$ is composed of terms in $A$ and terms not in $A$.

        Similarly $D$ is composed of terms in $B$ and terms not in $B$.

        There are 3 cases for this opposite form term existing in regards to $A$ and $B$:
        \begin{enumerate}
            \item the opposite form term does not exist in $A$ or $B$
            \item the opposite form term exists in either $A$ or $B$ but not both
            \item the opposite form term exists in $A$ and $B$
        \end{enumerate}
        
        Consider case (1) the opposite form term does not exist in A or B.

        Then we have the clauses:

        $A := [a, b, \beta]$

        $B := [c, d, \delta]$

        $C := [a, b, \beta, e, f, \phi]$

        $D := [c, d, \delta, -e, h, \gamma]$

        $E := [a, b, \beta, f, \phi, c, d, \delta, h, \gamma]$

        Notice all of the terms in $A$ exist in $E$ so we can expand $A$ to $E$ by Lemma 5.8.

        Since the length of $A$ is given as less than $k$, we can derive $E$ by processing clauses with a maximum length of $k - 1$.

        Consider case (2) the opposite form term exists in either A or B.

        Since $C$ and $D$ are treated the same, let's say the opposite form term exists in $D$ and $A$.
        
        Now we have the clauses:

        $A := [a, b, \beta]$

        $B := [c, d, \delta]$

        $C := [a, b, \beta, e, f, \phi]$

        $D := [c, d, \delta, -a, h, \gamma]$

        $E := [b, \beta, e, f, \phi, c, d, \delta, h, \gamma]$

        Notice all of the terms in $B$ exist in $E$ so we can expand $B$ to $E$ by Lemma 5.8.

        Since the length of $B$ is given as shorter than $k$, we can derive $E$ by processing clauses with a maximum length of $k - 1$.

        Consider case (3) the opposite form term exists in A and B.

        In this case, we redefine the clauses as follows:

        $A := [a, b, \beta]$

        $B := [-a, d, \delta]$

        $C = [a, b, \beta, e, f, \phi]$

        $D = [-a, d, \delta, g, h, \gamma]$

        $E = [b, \beta, e, f, \phi, d, \delta, g, h, \gamma]$

        In this case, by Lemma 5.9, $A$ and $B$ can imply a new clause:

        $F = [b, \beta, d, \delta]$

        Want to show $F$ is shorter than $k$.

        Consider two cases (3a) the length of $E$ is $k$ and (3b) the length of $E$ is $k-1$.

        Consider (3a) the length of $E$ is $k$.

        We can define $k$ as follows:

        $k = b + e + f + d + g + h 
            + \beta + \phi + \delta + \gamma
            - (\beta \phi) - (\beta \delta) - (\beta \gamma) - (\phi \delta) - (\phi \gamma) - (\delta \gamma)
            + (\beta \phi \delta) + (\beta \phi \gamma) + (\phi \delta \gamma)
            - (\beta \phi \delta \gamma)
            $

        And the length of $F$ is:

        Length of $F = b + d + \beta + \delta - (\beta \delta)$

        Want to show the length of $F$ is less than $k$:

        $b + d + \beta + \delta - (\beta \delta)$
        $<$
        $b + e + f + d + g + h 
            + \beta + \phi + \delta + \gamma
            - (\beta \phi) - (\beta \delta) - (\beta \gamma) - (\phi \delta) - (\phi \gamma) - (\delta \gamma)
            + (\beta \phi \delta) + (\beta \phi \gamma) + (\phi \delta \gamma)
            - (\beta \phi \delta \gamma)
        $

        $\rightarrow$
        $0$
        $<$
        $e + f + g + h 
            + \phi + \gamma
            - (\beta \phi) - (\beta \gamma) - (\phi \delta) - (\phi \gamma) - (\delta \gamma)
            + (\beta \phi \delta) + (\beta \phi \gamma) + (\phi \delta \gamma)
            - (\beta \phi \delta \gamma)
        $

        Which is true as long as at least one term exists on the R.H.S.

        Want to show at least one term exists on the R.H.S.

        Suppose not, then no terms exist on the R.H.S.

        Recall we have the clauses

        $A := [a, b, \beta]$

        $C := [a, b, \beta, e, f, \phi]$

        Since the terms on the R.H.S. do not exist, we can redefine the clause:

        $C := [a, b, \beta]$

        Notice no more terms may exist in $C$ because any new term could be treated as $e$ or $f$ and we know those don't exist.

        Notice $C$ is exactly $A$.

        This is a contradiction because the length of $A$ is given as less than $k$ and the length of $C$ is given as $k$.

        Therefore at least one term must exist on the R.H.S. and the inequality is true.

        Therefore $F$ is shorter than $k$ when $E$ is of length $k$.

        Consider (3b) the length of $E$ is $k - 1$.

        We can define $k$ as follows:

        $k = b + e + f + d + g + h 
            + \beta + \phi + \delta + \gamma
            - (\beta \phi) - (\beta \delta) - (\beta \gamma) - (\phi \delta) - (\phi \gamma) - (\delta \gamma)
            + (\beta \phi \delta) + (\beta \phi \gamma) + (\phi \delta \gamma)
            - (\beta \phi \delta \gamma)
            + 1
            $

        And the length of $F$ is:

        Length of $F = b + d + \beta + \delta - (\beta \delta)$

        Want to show the length of $F$ is less than $k$:

        $b + d + \beta + \delta - (\beta \delta)$
        $<$
        $b + e + f + d + g + h 
            + \beta + \phi + \delta + \gamma
            - (\beta \phi) - (\beta \delta) - (\beta \gamma) - (\phi \delta) - (\phi \gamma) - (\delta \gamma)
            + (\beta \phi \delta) + (\beta \phi \gamma) + (\phi \delta \gamma)
            - (\beta \phi \delta \gamma)
            + 1
        $

        Similarly to before, this will derive

        $\rightarrow$
        $0$
        $<$
        $e + f + g + h 
            + \phi + \gamma
            - (\beta \phi) - (\beta \gamma) - (\phi \delta) - (\phi \gamma) - (\delta \gamma)
            + (\beta \phi \delta) + (\beta \phi \gamma) + (\phi \delta \gamma)
            - (\beta \phi \delta \gamma)
            + 1
        $

        Which is always true.

        Therefore $F$ is shorter than $k$ when the length of $E$ is $k - 1$.

        Since we can derive $E$ by processing clauses with a maximum length of $k - 1$ for all possible cases of the lemma, then the lemma holds.
    \end{proof}

    \section{Algorithm}

    \begin{enumerate}
        \item For each clause in the instance, C, of length 3 or less:
        \begin{enumerate}
            \item For each clause in the instance, D, of length 3 or less:
            \begin{enumerate}
                \item Get all clauses implied by C and D according to Lemma 5.9 
                and add them to the instance
                \item Check if this new clause is in the instance and update a flag accordingly
            \end{enumerate}
            \item Expand C to get all possible clauses with a maximum length of 3 and add them to the instance
            \item For each new clause from the previous step
            \begin{enumerate}
                \item Check if the new clause is in the instance
            \end{enumerate}
        \end{enumerate}
        \item For each clause in the instance, E, of length 1:
        \begin{enumerate}
            \item For each clause in the instance, F, of length 1:
            \begin{enumerate}
                \item if E and F contain the same terminal in which it is 
                positive in one clause and negated in the other, the 
                clauses are contradicting and the instance is unsatisfiable, end
            \end{enumerate}
        \end{enumerate}
        \item Repeat (1)-(2) until no new clauses are added
        \item If it reaches here, the instance is satisfiable, end
    \end{enumerate}

    \section{Time Complexity Analysis}

    In this section I will analyze the time complexity of the algorithm in section 4

    (1) - $O(n^3)$ - At most ${\binom{n}{3}}*8 + {\binom{n}{2}}*4 + {\binom{n}{1}}*2$ clauses to iterate which is on the order of $O(n^3)$

    (1.a) - $O(n^3)$ -  At most ${\binom{n}{3}}*8 + {\binom{n}{2}}*4 + {\binom{n}{1}}*2$ clauses to iterate which is on the order of $O(n^3)$

    (1.a.i) - $O(1)$ - For each terminal in C, check if it's opposite form is in D. Since each clause is of length 3 or less, the worst case we iterate over 3 terms in C and check each term in D. If it's a match, we iterate over each clause again and create a new clause as described in Lemma 5.9. This time complexity is intuitively $O(3^2 + 3^2)$ which is constant time.

    (1.a.ii) - $O(n^3)$ - For each clause in the instance, iterate through the instance to check if it exists already. There are on the order of $O(n^3)$ clauses in the instance.

    (1.b) - $O(n^2)$ - For a 2-terminal clause, there are $2*n$ possible 3-terminal clauses that could be expanded to. For a 1-terminal clause, there are $4*n^2$ possible 3-terminal clauses that could be expanded to. This is upper bounded by the latter case.

    (1.c) - $O(n^2)$ - Upper bounded by at most $O(n^2)$ new clauses from the last step

    (1.c.i) - $O(n^3)$ - At most ${\binom{n}{3}}*8 + {\binom{n}{2}}*4 + {\binom{n}{1}}*2$ clauses to iterate which is on the order of $O(n^3)$

    (2) - $O(n^3)$ - Iterating through an instance where there are on the order of $O(n^3)$ clauses

    (2.a) - $O(n^3)$ - Iterating through an instance where there are on the order of $O(n^3)$ clauses

    (2.b.i) - $O(1)$ - Constant time to check if two 1-terminal clauses contain the same terminal in the opposite form

    (3) - $O(n^3)$ - Worst case, we add one new clause each time so we have to loop ${\binom{n}{3}} * 8 + {\binom{n}{2}} * 8 + {\binom{n}{1}} * 2$ times which is on the order of $O(n^3)$

    (4) - $O(1)$ - Constant time to check and return satisfiable

    The time complexity breaks down:

    $(3) * ((1) * ((1.a) * ((1.a.i) + (1.a.ii)) + (1.b) + (1.c) * (1.c.i)) + (2) * (2.a) * (2.a.i)) + (4)$

    It is seen the most computationally expensive steps are 

    $(3) * (1) * (1.a) * (1.a.oii)$

    $ = O(n^3) * O(n^3) * O(n^3) * O(n^3)$

    $= O(n^{12})$

    \section{Proof of Correctness}

    Want to show an instance is unsatisfiable iff we can derive contradicting
    1-terminal clauses by the algorithm.

    WTS Contradicting 1-terminal clauses can be derived $\implies$ the instance is unsatisfiable

    Contradicting 1-terminal clauses take the form:

    $[a]$

    $[-a]$

    Where $a$ is a terminal in the problem.

    In all possible assignments, $a$ can have the value of True or False.

    If $a$ is True, the clause $[-a]$ will be False and the assignment does not satisfy the instance.
    If $a$ is False, the clause $[a]$ will be False and the assignment does not satisfy the instance.

    Since all possible assignments do not allow both clauses to be True, the instance is unsatisfiable.

    Therefore contradicting 1-terminal clauses can be derived $\implies$ the instance is unsatisfiable

    Want to show an unsatisfiable instance $\implies$ the algorithm will derive contradicting 1-terminal clauses

    By Lemma 5.14, since the instance is unsatisfiable, the given 3-terminal clauses can be expanded to every possible $n$-terminal clause.

    By Lemma 5.15 these $n$-terminal clauses can be reduced by Lemma 5.7 to derive contradicting 1-terminal clauses.

    So we know if we can derive these $n$-terminal clauses, we can derive contradicting 1-terminal clauses.

    The idea behind the proof is that we know the $n$-terminal clauses can be used to derive the $1$-terminal via reduction but we'll show that we don't ever have to process a clause above length 3 to derive these contradicting 1-terminal clauses.

    This relies on the fact that all clauses of length 4 or greater, say of length $k$, have to be derived by shorter clauses and we can process the shorter clauses to derive any clauses that the clauses of length $k$ would derive.

    We can do this without processing a clause of length $k$ or greater.

    The way in which we derive these 1-terminal clauses is we use Lemma 5.7 to reduce the $n$-terminal clauses to $(n-1)$-terminal clauses which are reduced to $(n-2)$-terminal clauses which are reduced to ... which are reduced to $2$-terminal clauses and which finally get reduced to $1$-terminal clauses.

    Notice this passes through every possible $k$ from length $2$ to $n$.

    Recall that deriving all of the $n$-terminal clauses was done by using Lemma 5.8.

    These $n$-terminal clauses were then used to derive clauses of length $(n - 1)$ by Lemma 5.7.

    By Lemma 5.18, such a case allows us to derive the clauses of length $(n - 1)$ without ever having to process a clause of length $n$.

    Now we have all of the clauses of length $(n - 1)$ that we would have derived if we processed clauses of length $n$.

    Note how each of these clauses are either derived from given clauses by Lemma 5.8 or by Lemma 5.9 (all Lemma 5.7 derivations are a subset of all Lemma 5.9 derivations).

    We now want to use these $(n-1)$-terminal clauses to derive clauses of length $(n-2)$, but we want to do it without processing clauses whose length is greater than $(n-2)$.

    Since it takes two $(n-1)$-terminal clauses to derive a $(n-2)$-terminal clause by Lemma 5.7, the possible clauses could be of the form:

    \begin{enumerate}
        \item both clauses were derived by Lemma 5.8 (expansion)
        \item both clauses were derived by Lemma 5.9 (reduction/implication)
        \item each clause was derived in a different manner
    \end{enumerate}

    Note that this list is exhaustive because these are the only manner of implications used in the lemmas that allowed us to derive these clauses.

    We can use the following lemmas to handle each case:

    \begin{enumerate}
        \item Lemma 5.19
        \item Lemma 5.17
        \item Lemma 5.18
    \end{enumerate}

    As such we can derive the clauses of length $(n - 2)$ without processing a clause whose length is greater than $(n - 2)$.

    In a more general sense, for any implied clause of length $k$, say $C$, that is used to derive a clause of length $k$ or $k - 1$, say $D$, then we can use the fact that $C$ is derived by shorter clauses and we can use these shorter clauses to directly derive $D$ without ever processing a clause of length $k$ or greater.

    At this point, we have 4-terminal clauses and we want to derive 3-terminal clauses. 
    
    Notice that Lemmas 5.17, 5.18, and 5.19 require the input clauses to be derived so we cannot use those lemmas to derive all the clauses of length 3 we need.

    We need two input clauses to derive the necessary 3-terminal clauses and we have three cases:

    \begin{enumerate}
        \item Both inputs are of length 4
        \item One input is of length 4, the other is of length 3
        \item Both inputs are of length 3
    \end{enumerate}

    We can disregard the last point because we want to derive a 3-terminal clause without processing a clause of length 4 or greater so the claim is vacuously true in this case.

    In the case where both inputs are of length 4, we know they are both derived using smaller clauses and we can use Lemmas 5.17, 5.18, or 5.19.

    In the case where one input is of length 4 and the other is of length 3, there are two cases for how the 4-terminal clause was derived:

    \begin{enumerate}
        \item It was derived using Lemma 5.9 (reduction/implication)
        \item It was derived using Lemma 5.8 (expansion)
    \end{enumerate}

    We can handle the cases in the following ways:

    \begin{enumerate}
        \item Using Lemma 5.11
        \item Using Lemma 5.12
    \end{enumerate}

    Now we have all the 3-terminal clauses that would have been derived while the $n$-terminal clauses were being reduced to contradicting 1-terminal clauses.

    We can now reduce the 3-terminal clauses to 1-terminal clauses.

    Since we have all the necessary 3-terminal clauses without having to process a clause of length 4 or greater, this shows we do not have to process any clauses of 4 or greater to derive contradicting 1-terminal clauses.
    
    Even though the algorithm only explicitly uses Lemma 5.9 and Lemma 5.8, this will cover cases where reduction is needed because Lemma 5.9 is a more general case of Lemma 5.7. Notice, too, that Lemmas 5.11, 5.12, 5.17, 5.18, and 5.19 rely on Lemmas 5.8 and 5.9 so we do not have to explicitly capture the cases where the intermediate lemmas would apply.

    Therefore, this coincides with the described algorithm.

    \section{Conclusion} 

    In this paper, we present an algorithm to solve 3SAT in polynomial time.

    This algorithm relies on Lemmas 5.8 and 5.9. The former of which says a clause can expand to imply another clause by appending any term that's not in the original clause. The latter of which says two clauses sharing one terminal which is positive in one clause and negated in the other can imply a new clause composed of the rest of the terms in either clause. Using these strategies, we can process an instance of 3SAT  while only considering clauses of length 3 or less and are guaranteed to derive a pair of contradicting 1-terminal clauses if and only if the instance is unsatisfiable.

    According to \cite{Karp1972}, 3SAT is NP-complete so if 3SAT can be solved in polynomial time then every problem in NP can be solved in polynomial time.
    
    Since such an algorithm exists, all problems in NP can be solved in polynomial time.

    Thus, P = NP. 

    \bibliography{refs}
    \bibliographystyle{plain}

\end{document}